\documentclass[12pt,reqno]{amsart}
\usepackage[left=3.2cm, top=3.4cm, bottom=3.4cm, right=3.2cm]{geometry}
\usepackage{mathrsfs}
\usepackage{amsmath,amsxtra,amsfonts,amssymb, amsthm, amscd, epsfig}
\usepackage[dvipsnames,usenames]{color}
\usepackage{amssymb,amsmath,amsfonts,epsfig}
\usepackage{graphicx}

\numberwithin{equation}{section}

\theoremstyle{plain}
\newtheorem{thm}{Theorem}[section]

\newtheorem{cor}[thm]{Corollary}

\newtheorem{lem}[thm]{Lemma}
\newtheorem{de}[thm]{Definition}

\newtheorem{ex}[thm]{Example}
\newcommand{\eqa}{\begin{eqnarray}}
\newcommand{\eeqa}{\end{eqnarray}}
\newcommand{\beq}{\begin{equation}}
\newcommand{\eeq}{\end{equation}}
\newcommand{\nn}{\nonumber}
\newcommand{\p}{\partial}
\def \be{{\bf 1_n}}
\def \la {\langle}
\def \ra{\rangle}
\def \var{\varepsilon}
\def \trm{\mathrm{tr}\,}

\def \res{\mathrm{res\,}}
\def \dsum{\displaystyle\sum}

\begin{document}

\title[]
{Integrability of the Frobenius algebra-valued KP hierarchy}
\author[]{Ian A.B. Strachan and Dafeng Zuo}

\address{Ian A.B. Strachan,School of Mathematics and Statistics, University of Glasgow}
\email{ian.strachan@glasgow.ac.uk}

\address[]{Dafeng Zuo,School of Mathematical Science, University of Science and
Technology of China, Hefei 230026, P.R.China and Wu Wen-Tsun Key Laboratory of Mathematics,
USTC, Chinese Academy of Sciences and School of Mathematics}

\email{dfzuo@ustc.edu.cn}

\date{\today}

\begin{abstract}
We introduce a Frobenius algebra-valued KP hierarchy
and show the existence of Frobenius algebra-valued $\tau$-function for this hierarchy. In addition we construct its Hamiltonian structures by using the Adler-Dickey-Gelfand method. As a byproduct of these constructions, we show that the coupled KP hierarchy defined by P.Casati and G.Ortenzi in \cite{CO2006} has at least $n$-``basic" different local bi-Hamiltonian structures.
 Finally, via the construction of the second Hamiltonian structures, we  obtain some local matrix, or Frobenius algebra-valued, generalizations
  of classical $W$-algebras.

\end{abstract}


\maketitle 
{ \tableofcontents}
\section{Introduction}

The Kadomtsev-Petviashvili (KP) hierarchy is defined by the set of equations
\beq \dfrac{\p}{\p t_r}L=[B_r,L], \quad r=1,2,\cdots, \label{Z0.2}\eeq
where $L=\p+u_1\p^{-1}+u_2\p^{-2}+\cdots$ is a pseudo-differential operator
with coefficients $u_1$, $u_2$, $\cdots$ being smooth functions of
infinitely many variables ${\bf t}=(t_1,t_2,\cdots)$ with $t_1=x$
and $B_r=L^r_+$ is the pure differential part of the operator $L^r$ and $\p=\frac{\p}{\p x}$.

A fundamental result, due to M.Sato, is the existence of
a $\tau$-function for the KP hierarchy (see the survey \cite{DJKM1983}).
Another fundamental property of this hierarchy is that it has two compatible local Hamiltonian structures. The
first structure was suggested by Watanabe \cite{W1983}, the second by Dickey \cite{D1987} and,
 shortly after that, Radul \cite{R1987} proved that not only one pair of structures can be
built but infinitely many. Essentially, the construction
was a slight modification of the Adler-Gelfand-Dickey (AGD) method for the
$n^{\it th}$-KdV $(\mbox{GD}_n)$ hierarchy  in \cite{A1979,GD1978}. We refer to
\cite{Dick2003} for a more detailed description.

Recently, there are several types of noncommutative generalizations of
the KP hierarchy (see, for example, \cite{Kup2000} and the references therein).
Most of them do not preserve the above two fundamental properties. For example,
the matrix KP \cite{Bial1995} has two compatible Hamiltonian structures via
the AGD method, on utilizing the matrix trace map, but
the second is nonlocal. Furthermore,  there is no $\tau$-function for this hierarchy.

In this paper we study certain properties of a Frobenius algebra-valued KP hierarchy.
The first motivation stems from the work of Casati and Ortenzi \cite{CO2006}.  With the use of
vertex operator representations of polynomial Lie algebras,
they obtained a class of coupled KP hierarchy formulated as a ``coupled Hirota bilinear equation". Shortly afterwards,
Van de Leur in \cite{L2007}, starting from these bilinear equations,
recovered the corresponding  wave functions and Lax equations with
a $\mathcal{Z}_n$-valued Lax operator $L$, where $\mathcal{Z}_n=\mathbb{C}[\Lambda]/(\Lambda^{n})$
is the maximal commutative subalgebra of $gl(n,\mathbb{R})$ and $\Lambda=(\delta_{i,j+1})\in gl(m,\mathbb{R})$.
A natural problem then arises as to how to construct Hamiltonian structures for these coupled KP
hierarchies. The main problem in using the AGD method is that the bilinear form constructed using the usual matrix trace
$\la A,B\ra=\mbox{trace}\, (AB)$ for $A,B\in \mathcal{Z}_n$ is degenerate.
In order to solve this, one of the current authors \cite{Zuo-2013}
introduced a somewhat strange looking trace-type map
$\trm_n: gl(n,\mathbb{R})\longrightarrow\mathbb{R}$
defined by
\beq \trm_n(A)=\mbox{trace}~\left[\left(\begin{array}{ccccc}
\frac{1}{n}& \frac{1}{n-1}&\cdots & 1\\
0&\frac{1}{n}&\ddots &\frac{1}{2}\\
\vdots&\vdots&\ddots&\vdots\\
0&0&\cdots&\frac{1}{n}
\end{array}\right)A\right].\label{DF1.5}\eeq
We remark that this trace-type map is not symmetric on $gl(n,\mathbb{R})$ but when restricted to the subalgebra
$\mathcal{Z}_n$, is nondegenerate and symmetric.

Our second motivation is due to the following crucial observation. Let $\bf{1}_n$ be the identity matrix and $\circ$  the matrix multiplication, then $\{\mathcal{Z}_n, \mathrm{tr}_n, \bf{1}_n, \circ\}$
is a Frobenius algebra. This observation motivates us to study the $\mathcal{A}$-valued KP hierarchy via $\mathcal{A}$-valued
Lax operators, where $\mathcal{A}$ is a Frobenius algebra.

This paper is organized as follows. In section 2, we will show the existence of the $\mathcal{A}$-valued $\tau$-function for the $\mathcal{A}$-KP hierarchy. In section 3, we will construct Hamiltonian structures of the $\mathcal{A}$-valued KP hierarchy.
In section 4, we will list some similar results for the $\mathcal{A}$-valued dispersionless KP hierarchy.
Section 5 is devoted to various conclusions and a discussion of some open problems.


\section{The Frobenius algebra-valued KP hierarchy and its $\tau$-function}

In this section, we will introduce an $\mathcal{A}$-valued KP hierarchy via $\mathcal{A}$-valued
Lax operators and show the existence of an $\mathcal{A}$-valued $\tau$-function for the $\mathcal{A}$-KP hierarchy.

\subsection{Frobenius algebra} We begin with the definition of a Frobenius algebra \cite{Du}.

\begin{de}
A Frobenius algebra $\{ \mathcal{A},\circ,e,\omega\}$ over $\mathbb{R}$ satisfies the following conditions:

\begin{itemize}

\item[{\sl (i)}] $\circ\,: \mathcal{A} \times \mathcal{A} \rightarrow \mathcal{A}$ is
a commutative, associative algebra with unity $e$;

\item[{\sl (ii)}] $\omega \in \mathcal{A}^\star$ defines a non-degenerate inner product
$\langle a,b \rangle = \omega(a \circ b)\,$, which is often called a trace form (or Frobenius form).

\end{itemize}
\end{de}

\begin{ex}(\cite{SZ1}) \label{ex2.2}
Let $\mathcal{A}$ be a two-dimensional commutative and associative algebra with a basis $e=e_1, e_2$
satisfying
\beq e_1\circ e_1=e_1,\quad e_1\circ e_2=e_2, \quad e_2\circ e_2=\var e_1+\mu e_2,\quad \var,\mu \in \mathbb{R},
 \eeq
then $\mathcal{Z}_{2,k}^{\var,\mu}:=\left\{\mathcal{A},\circ, e, \omega_k\right\}$, $k=1,2$
are Frobenius algebras, where
\beq \omega_k(a)=a_k+a_2(1-\delta_{k,2})\delta_{\var,0}, \quad k=1,\, 2,\label{ZZ2.4}\eeq
for $a=a_1e_1+a_2e_2\in\mathcal{A}$.
\end{ex}

\begin{ex}(\cite{SZ1}) \label{ex2.3}
Let $\mathcal{A}$ be an $n$-dimensional  nonsemisimple commutative associative
algebra $\mathcal{Z}_n$ over $\mathbb{R}$ with a unity $e$ and a basis $e_1=e, \cdots,
e_n$ satisfying
\beq e_i \circ e_j=\left\{
\begin{array}{ll}
e_{i+j-1}, & i+j\leq n+1,\\
0,&  i+j=n+2. \end{array}\right.
 \eeq
Taking $\Lambda=(\delta_{i,j+1})\in gl(m,\mathbb{R})$, one obtains
a matrix representation of $\mathcal{A}$ as
 \beq e_j \mapsto \Lambda^{j-1}, \quad j=1,\cdots, n.\nn \eeq
Similarly, for any $a=\dsum_{k=1}^{n} a_k e_k\in \mathcal{A}$,
we introduce $n$ trace-type forms, called ``basic" trace-type forms,  as follows
\beq \omega_{k-1} (a)=a_k+a_{n}(1-\delta_{k,n}),\quad k=1,\cdots,n.\label{ZZ2.3} \eeq
Every trace map $\omega_k$ induces a nondegenerate
symmetric bilinear form on $\mathcal{A}$ given by
 \beq \la a, b\ra_k:=\omega_k (a\circ b),\quad a,\,b\in \mathcal{A},\quad k=0,\cdots,n-1.\label{AZ2.6}\eeq
Thus all of $\left\{\mathcal{A}, \circ, e, \omega_{k-1}\right\}$ are nonsemisimple
Frobenius algebras, denoted by $\mathcal{Z}_{n,k-1}$ for $k=1,\cdots,n$.
We remark that the trace-type map $\mathrm{tr}_n$ in \eqref{DF1.5} is exactly
a linear combination of $n$ ``basic" trace-type forms as
\beq \mathrm{tr}_n:=\dsum_{s=0}^{n-1}\omega_{s}-(n-1)\,\omega_{n-1}.\nn\eeq
\end{ex}

\noindent Unless otherwise stated, we assume that $\left\{\mathcal{A},\circ, e:=\be, \omega:=\mathrm{tr}\right\}$ is an
$n$-dimensional Frobenius algebra over $\mathbb{R}$ with the basis $e_1=\be, e_2\cdots,e_n$.

\subsection{The $\mathcal{A}$-valued KP hierarchy}
Let
\beq L=\be\p+U_1\p^{-1}+U_2\p^{-2}+\cdots, \quad \p=\dfrac{\p}{\p x},\label{SZ2.2}\eeq
be an $\mathcal{A}$-valued pseudo-differential operator \rm{($\Psi$DO)} with coefficients $U_1$, $U_2$, $\cdots$ being smooth
$\mathcal{A}$-valued functions of an infinite many variables $t=(t_1,t_2,\cdots)$ and $t_1=x$.
\begin{de}  The $\mathcal{A}$-KP hierarchy is  the set of equations
\beq \dfrac{\p L}{\p t_r}=[B_r,L]:=B_r\circ L-L\circ B_r, \quad B_r=L^r_{+}, \quad r=1,2,\cdots, \label{SZ2.3}\eeq
where $B_r$ is the pure differential part of the operator $L^r=\underbrace{L\circ \cdots \circ L}_{r\, \rm{terms}}$.
\end{de}

Generally, by imposing the constraint $(L^m)_{-}=0$, the
$\mathcal{A}$-KP hierarchy \eqref{SZ2.3} reduces to the $\mathcal{A}$-$\mbox{GD}_m$ hierarchy.
The $\mathcal{A}$-KP hierarchy is equivalent to
 \beq \dfrac{\p B_l}{\p t_r}-\dfrac{\p B_r}{\p t_l}+[B_l,B_r]=0. \label{Z2.4} \eeq
Consider a case ($r=2$,\, $l=3$), then the system \eqref{Z2.4} becomes
\beq
 U_{1,t_2}=U_{1,xx}+2U_{2,x},\quad 2 U_{1,t_3}=2 U_{1,xxx}+3\,U_{2,xx}+3\,U_{2,t_2}+6\,U_1\circ U_{1,x}.\label{Z2.5}
\eeq
If we eliminate $U_2$ in \eqref{Z2.5} and rename $t_2=y$, $t_3=t$ and $\mathcal{U}=U_1$, we obtain
\beq (4\mathcal{U}_t-12\mathcal{U}\circ \mathcal{U}_x-\mathcal{U}_{xxx})_x-3\mathcal{U}_{yy}=0.\label{Z2.6}\eeq
All this follows the scalar case verbatim. But as the following example shows, when written in terms of a specific basis
this structure is broken and the underlying Frobenius algebra is hidden.

\begin{ex}
Suppose that $\mathcal{A}$ is the $\mathcal{Z}_{2,2}^{\var,\mu}$ algebra and
$\mathcal{U}=v e_1+w e_2$. Then the system \eqref{Z2.6} in component form is
\beq \left\{\begin{array}{l}
(4v_{t}-12vv_{x}-v_{xxx}-12\var\,ww_{x})_x-3v_{yy}=0,\\
(4w_{t}-12(vw)_x-w_{xxx}-12\mu\,ww_{x})_x-3w_{yy}=0.
\end{array}\right.\label{Z2.7} \eeq
When $\var=\mu=0$, the system \eqref{Z2.7} reduces to the coupled KP equation (e.g.\cite{CO2006,Zuo-2013}).
Furthermore, if $v_y=w_y=0$, the coupled KP equation reduces to the coupled KdV equation
(\cite{FR1989,MF1996, HXT2003})
%
\beq \left\{\begin{array}{l}
4v_{t}-12vv_{x}-v_{xxx}=0,\,\\ 4w_{t}-12(vw)_x-w_{xxx}=0.
\end{array}\right.\label{CKDV}\eeq
\end{ex}

\noindent Thus certain multicomponent examples that have appeared in the literature are best viewed as a single $\mathcal{A}$-valued
equation: writing them in terms of basis-dependent component fields obscures the underlying algebraic structure.

\subsection{The $\tau$-function}
Let us represent the $\mathcal{A}$-valued Lax operator $L$ in \eqref{SZ2.2} in a dressing form
\beq L=\Phi^{-1}\circ \be \p\circ \Phi, \quad \mbox{$\Phi=\dsum_{i=0}^\infty W_i \p^{-i}$ with $W_0=\be$,}\label{A.1} \eeq
where the $\mathcal{A}$-valued dressing operator $\Phi$ is determined up to a multiplication
on the right
by $\be+ \dsum_{k=1}^\infty C_k \p^{-k}$ with arbitrary constant elements $C_k\in\mathcal{A}$.
Then using \eqref{SZ2.3}, we obtain
\beq \p_r\Phi=-L_{-}^r\circ\Phi, \qquad \p_r=\frac{\p}{\p t_r}.\label{A.2}\eeq
For simplicity, let $\xi(t,z)=\dsum_{k=1}^\infty t_kz^k$ and $\widehat{W}(t,z)=\dsum_{i=0}^\infty{W}_iz^i$, where
$z\in\mathbb{C}$ is a parameter.
The wave function of the $\mathcal{A}$-KP hierarchy \eqref{SZ2.3} is defined by
the $\mathcal{A}$-valued function
\beq W(t,z):=\Phi e^{\xi(t,z)}=\widehat{W}(t,z)e^{\xi(t,z)}.\label{A.4}\eeq
Similarly, the adjoint wave function is given by
\beq \widetilde{W}(t,z):=(\Phi^{-1})^* e^{-\xi(t,z)}=\widehat{\tilde{W}}(t,z)e^{\xi(t,z)}.\eeq

\begin{lem} The following  identities hold:
\eqa &(1).& \quad \mathrm{res}_z\big(\p^{i_1}_1\cdots \p^{i_k}_k W(t,z)\big) \circ\widetilde{W}(t,z)=0,\quad i_j\in \mathbb{Z}_{\geq 0}; \label{A.7} \\
&(2).&\quad \widehat{W}(t,z)^{-1}=G(z)[\widehat{\tilde{W}}(t,z)];\label{A.8}\\
&(3).& \quad \p \ln\widehat{W}(t,z)=W_1(t)-G(z)[W_1(t)],\label{A.9}
\eeqa
where $G(z)$ is a shift operator defined by
\beq G{(z)}[f(t;z,s)]=f(t_1-\frac{1}{z}, t_2-\frac{1}{2z^2},\cdots;z,s). \label{A.6}\eeq
The identity \eqref{A.7} is called the $\mathcal{A}$-valued bilinear identity. \end{lem}
\begin{proof} (1). By definitions,
\beq L^r\circ W(t,z)=(\Phi \circ \p \Phi^{-1})^r \circ\Phi e^{\xi(t,z)}=z^r W(t,z). \label{A.10}\eeq
Using \eqref{A.2} and \eqref{A.8}, we have
\beq \p_r W(t,z)=\p_r(\Phi e^{\xi(t,z)})=(\p_r\Phi) e^{\xi(t,z)}+\Phi (\p_re^{\xi(t,z)})=L_+^rW(t,z).\label{A.11}\eeq
With \eqref{A.9}, it suffices to consider only the case when all $i_j$ for $j>1$ vanish. Then
\eqa \mathrm{res}_z\big(\p^i W(t,z)\big) \circ\widetilde{W}(t,z)=
\mathrm{res}_z\big(\p^i \Phi e^{xz}\big) \circ\big((\Phi^{-1})^* e^{-xz}\big)
=\mathrm{res}_{\p}\p^i\Phi\circ\Phi^{-1}=0.\nn\eeqa
In the second step, we use a simple formula $\mathrm{res}_z (P e^{xz}\big) \circ\big(Q e^{-xz}\big)=\mathrm{res}_{\p}P\circ Q^*$, where $P$ and $Q$ are two $\mathcal{A}$-valued $\Psi$DOs.

(2). The bilinear identity \eqref{A.7} implies
\beq \mathrm{res}_z W(t,z)\circ G(\zeta)[\widetilde{W}(t,z)]=0.\label{A.12}\eeq
Using \eqref{A.12} and the identity $e^{\sum_{k=1}^\infty \frac{z^k}{k\zeta^k}}=(1-\frac{z}{\zeta})^{-1}$,
we obtain
\beq 0=\mathrm{res}_z \widehat{W}(t,z)\circ G(\zeta)[\widehat{\tilde{W}}(t,z)](1-\frac{z}{\zeta})^{-1}
=\zeta\big(\widehat{W}(t,\zeta)\circ G(\zeta)[\widehat{\tilde{W}}(t,z)]-\be\big)\nn \eeq
which yields the identity \eqref{A.8}.

(3). Similarly, from the bilinear identity \eqref{A.7} we have
$$\mathrm{res}_z \p W(t,z)\circ G(\zeta)[\widetilde{W}(t,z)]=0.$$
So using \eqref{A.8}, we get
\eqa 0
&=&\mathrm{res}_z \p \widehat{W}(t,z)\circ G(\zeta)[\widehat{\tilde{W}}(t,z)](1-\frac{z}{\zeta})^{-1}\nn\\
&=&\big(\p \widehat{W}(t,\zeta)+\zeta \widehat{W}(t,\zeta)\big)\circ G(\zeta)[\widehat{\tilde{W}}(t,\zeta)]-\be \zeta-W_1(t)+G(\zeta)[W_1(t)]
\nn\\
&=& \p \widehat{W}(t,\zeta)\circ \widehat{W}(t,\zeta)^{-1}-W_1(t)+G(\zeta)[W_1(t)] \nn
\eeqa
which implies the identity \eqref{A.9}.
\end{proof}

We are now in a position to state the main theorem in the this section, which
can be regarded as an $\mathcal{A}$-valued counterpart of Sato's theorem
\cite{DJKM1983,D1997,Dick2003,SS1982} for the scalar KP hierarchy.

\begin{thm}There is an $\mathcal{A}$-valued function $\tau=\tau(t)$ such that
\beq \widehat{W}(t,z)= G(z)[\tau(t)]\circ\tau(t)^{-1}. \label{A.5}\eeq
The $\mathcal{A}$-valued  $\tau$-function is determined up to a multiplication
by $C_0\circ e^{\sum_{k=1}^\infty C_k t_k} $ with arbitrary constant elements $C_k\in\mathcal{A}, k\in \mathbb{N}$ and
arbitrary invertible constant element $C_0\in\mathcal{A}$.
\end{thm}
\begin{proof}With the bilinear identity \eqref{A.7}, we get
\beq \mathrm{res}_z W(t,z)\circ G(\zeta_1)[G(\zeta_2)[\widetilde{W}(t,z)]]=0\nn\eeq
and
\beq \mathrm{res}_z \widehat{W}(t,z)\circ G(\zeta_1)[G(\zeta_2)[\widehat{\tilde{W}}(t,z)]]
(1-\frac{z}{\zeta_1})^{-1}(1-\frac{z}{\zeta_2})^{-1}=0.\label{A.13}\eeq
It follows from \eqref{A.13} that
\beq \widehat{W}(t,\zeta_1)\circ G(\zeta_2)[G(\zeta_1)[\widehat{\tilde{W}}(t,\zeta_1)]]
=\widehat{W}(t,\zeta_2)\circ G(\zeta_1)[G(\zeta_2)[\widehat{\tilde{W}}(t,\zeta_2)]], \nn \eeq
which becomes, using \eqref{A.8},
\beq \widehat{W}(t,\zeta_1)\circ G(\zeta_2)[\widehat{W}(t,\zeta_1)^{-1}]
=\widehat{W}(t,\zeta_2)\circ G(\zeta_1)[\widehat{W}(t,\zeta_2)^{-1}]. \label{A.14}\eeq
Letting $\mu(t,z)=\ln \widehat{W}(t,z)$ and
taking into account \eqref{A.14},  we have
\beq \mu(t,\zeta_1)-G(\zeta_2)[\mu(t,\zeta_1)]=\mu(t,\zeta_2)-G(\zeta_1)[\mu(t,\zeta_2)].\label{A.15}\eeq
For simplicity, we denote
\beq N(z):=\frac{\p}{\p z}-\sum_{k=1}^\infty z^{-k-1}\p_k,\quad B_i:=\mathrm{res}_z z^iN(z)\mu(t,z).\nn\eeq
Applying the operator to \eqref{A.15} after renaming $\zeta_1=z$ and $\zeta_2=\zeta$, we get
\beq N(z) \mu(t,z)-G(\zeta)[N(z)\mu(t,z)]=-\sum_{k=1}^\infty z^{-k-1}\p_k\mu(t,\zeta).\label{A.16}\eeq
Multiplying by $z^i$ on both sides of \eqref{A.16} and taking the residues $\mathrm{res}_z$, we obtain
\beq B_i=G(\zeta)[B_i]-\p_i\mu(t,z) \label{A.17} \eeq
and furthermore,
\beq \p_j B_i-\p_i B_j=G(\zeta)[\p_j B_i-\p_i B_j].\label{A.18} \eeq
which yields $\p_j B_i-\p_i B_j=const\in\mathcal{A}$. The left side of \eqref{A.18}
is a differential polynomial in $W_i(t)$ without constant terms, we thus have
$\p_j B_i=\p_i B_j$. So there is an $\mathcal{A}$-valued function $\tau=\tau(t)$
such that $B_i=\p_i \ln \tau$. By using \eqref{A.17}, we get
\beq \p_i \mu(t,z)=\p_i(G(\zeta)[\ln \tau]-\ln \tau)\nn\eeq
which yields \eqref{A.5}. The rest of the theorem is obvious.  \end{proof}

\begin{cor} For any $i\in\mathbb{N}$, the following identity holds:
\beq \res L^i=\frac{\p}{\p t_i} (\tau_x\circ\tau^{-1}).\label{A.19}\eeq \end{cor}
\begin{proof} Equating the residue on both sides of \eqref{A.2}, we have
\beq  \res L^i=-\p_i W_1(t).\label{A.20}\eeq
Observe that $\mathrm{res}_z z^iN(z)\mu(t,z)=B_i=\p_i \ln \tau $ and $\mu(t,z)=\ln\widehat{W}(t,z)$,
then we get
\eqa \frac{\p}{\p t_i} (\tau_x\circ\tau^{-1})&=&\p\p_i\ln\tau=\mathrm{res}_z z^iN(z)\p \ln\widehat{W}(t,z)\qquad
\mbox{using \eqref{A.9}}\nn\\
&=& \mathrm{res}_z z^iN(z)\big(W_1(t)-G(z)[W_1(t)]\big)=\mathrm{res}_z z^iN(z)W_1(t)\nn\\
&=& \mathrm{res}_z z^i \dsum_{k=1}^\infty z^{-k-1}\p_kW_1(t)=-\p_iW_1(t).\nn
\eeqa
Taking into account \eqref{A.20}, we obtain the desired formula \eqref{A.19}.
\end{proof}

\begin{ex}
Let $A\in\mathcal{A}$ be a constant element, then $$\tau={\bf{1}}_n+\exp(2Ax+2A^3t)$$ is an
$\mathcal{A}$-valued $\tau$-function of the  $\mathcal{A}$-valued KdV equation
$ 4\mathcal{U}_t-12\mathcal{U}\circ \mathcal{U}_x-\mathcal{U}_{xxx}=0$.

Taking $\mathcal{A}$ to be the Frobenius algebra $\mathcal{Z}_{2}$,  the $\mathcal{A}$-valued KdV equation is exactly
the coupled \mbox{KdV} equations \eqref{CKDV}.
By choosing
$A=\left(\begin{array}{cc}
a&0\\
b&a
\end{array}\right)\in \mathcal{A}$, we then have
$$\tau=\left(\begin{array}{cc}
1+\exp(2ax+2a^3t)&0\\
(2bx+2b^3t)\exp(2ax+2a^3t)& 1+\exp(2ax+2a^3t)
\end{array}\right):=\left(\begin{array}{cc}
\tau_0&0\\
\tau_1&\tau_0
\end{array}\right)$$
and
\beq \left(\begin{array}{cc}
v&0\\
w&v
\end{array}\right)=\mathcal{U}=\frac{\p}{\p x}(\tau_x\tau^{-1})= \left(\begin{array}{cc}
(\log \tau_0)_{xx}&0\\
\left(\dfrac{\tau_1}{\tau_0}\right)_{xx}&(\log \tau_0)_{xx}
\end{array}\right).\nn\eeq
Thus we obtain a solution of the couple KdV equation \eqref{CKDV} given by
\beq  v=(\log \tau_0)_{xx},\quad w=\left(\dfrac{\tau_1}{\tau_0}\right)_{xx}. \label{ZZA1} \eeq
We remark that the variable transformation \eqref{ZZA1} has been used to
derive the coupled KdV equation from the Hirota equation in \cite{CO2006}.
The form of this (i.e. equation \eqref{ZZA1}) may thus be traced back to the nilpotent elements the appear in the
Frobenius algebra $\mathcal{A}.$
\end{ex}

\section{Hamiltonian structures of the $\mathcal{A}$-valued KP hierarchy}

In this section, we will use the AGD-scheme (e.g.\cite{A1979,GD1978,Dick2003}) to construct
Hamiltonian structures of the $\mathcal{A}$-\mbox{KP} hierarchy. For the clarity,
let $P=\dsum_i P_i\p^i$ be an $\mathcal{A}$-valued
\mbox{$\Psi$DO}, in what follows we donete $P_+$ the pure differential part of the operator $P$ and
$$P_{-}=P-P_+, \quad \mbox{res}(P)=P_{-1},\quad P^*=\dsum_i (-1)^i \p^i P_i.$$

\begin{lem}\label{lem2.4}
Suppose $A$ and $B$ are two $\mathcal{A}$-valued \mbox{$\Psi$DO}s, then
\beq \trm \int \res A\circ B\,dx=\trm\int \res B\circ A \,dx.\label{eq3.1}\eeq\end{lem}

\begin{proof}We first show that
\beq \res[A,B]= \displaystyle{\frac{\p h(x,{\bf t})}{\p x}},\label{eq3.2}\eeq
where $h(x,{\bf t})$ is a certain $\mathcal{A}$-valued function. By linearity, it is sufficient to
prove \eqref{eq3.1} for any two $\mathcal{A}$-valued
monomials $A=A_i\p^i$, $B=B_j\p^j$. If $i,j\geq 0$ or $i+j<1$, then
$\res[A,B]=0$ and so $h=0$. We thus only need consider the case $i\geq 0$, $j<0$ and $i+j\geq 1$.
A direct computation gives
\eqa \res[A,B]&=&C_i^{i+j+1}\left(A_i\circ B_j^{(i+j+1)}+(-1)^{i+j}B_j\circ A_i^{(i+j+1)}\right)\nn\\
&=&
\frac{\p}{\p x} \left(C_i^{i+j+1}\dsum_{s=0}^{i+j}(-1)^sA_i^{(s)}\circ B_j^{i+j-s}\right):=\frac{\p}{\p x}h.\nn
\eeqa
Obviously $h$ is $\mathcal{A}$-valued. Furthermore, taking the trace form $\trm$ on both sides of \eqref{eq3.2}, we obtain
$\trm \res[A,B]=\trm \frac{\p h}{\p x}$. With this, the identity \eqref{eq3.1} follows immediately.
\end{proof}

\subsection{Case $U_0\ne 0$, i.e., $V_{m-1}\ne 0$} Let $L=\be\p+U_0+U_1\p^{-1}+U_2\p^{-2}+\cdots$
be an $\mathcal{A}$-valued $\Psi$DO with an additional term $U_0$.
Denoting
\beq  \mathcal{L}:=L^m=\be \p^m+V_{m-1}\p^{m-1}+V_{m-2}\p^{m-2}+\cdots,
\quad V_{i}=\dsum_{q=1}^{n} v_{[i]q}e_q.\label{Z2.15}\eeq
In the following our Hamiltonian structures will be established in terms
of the ``dynamical coordinates" $\{v_{[i]q}\}$.

We denote by $ {\mathfrak{D}}$ the differential algebra of polynomials in formal symbols
$\left\{v_{[i]q}^{(j)}\right\}$, where $v_{[i]q}^{(j)}=\dfrac{\p^jv_{[i]q}}{\p x^j}$ for
$q=1,\cdots,n$ and $j=0,1,\cdots$.  We consider a subalgebra $\mathcal{D}$ of ${\mathfrak{D}}$ with the element of the form $\trm F(V)$, where $F(V)$ is an $\mathcal{A}$-valued
differential polynomial w.r.t. its arguments $V_i$.
 We denote the space of functionals by
$$\widetilde{\mathcal{D}}=\left\{\left.\tilde{f}=\int \trm F(V) dx\right|\,\trm F(V)\in {\mathcal{D}} \right\}.$$

The variational derivative with respect to an algebra-valued field has been discussed in \cite{OS} .
In the present context, for $V=\dsum_{q=1}^{n}v_qe_q$, the variational derivative $\dfrac{\delta F }{\delta V}$  is defined by
\beq \tilde{f}(v+\delta v)-\tilde{f}(v)=\int \trm \left(\dfrac{\delta F }{\delta V}\circ\delta V+o(\delta V)\right)dx
=\int\dsum_{q=1}^{n}\left(\dfrac{\delta f }{\delta v_q}\,\delta v_q+o(\delta v)\right)dx,\label{Z2.20} \eeq
where $f(v)= \trm F(V)$, $\delta V=\dsum_{q=1}^{n}\delta v_q e_q \in \mathcal{A}$ and $
\dfrac{\delta f}{\delta v_{q}}=\dsum_{j=0}^{\infty}(-\p)^j\dfrac{\p f}
{\p v_{q}^{(j)}}$. Without confusion, we use the notation
$\dfrac{\delta f }{\delta V}$ instead of $\dfrac{\delta F }{\delta V}$.

Suppose ${\bf a}=(a_{m-1},a_{m-2},\cdots)$ with elements
\beq a_i=\dsum_{q=1}^{n} a_{[i]q}e_q \in \mathcal{A},\quad i=m-1,m-2,\cdots\,.\nn\eeq
We define a vector field associated to ${\bf a}$ by the formula
\beq \p_{\bf a}=\dsum_{i=-\infty}^{m-1}\dsum_{j=0}^{\infty} \dsum_{q=1}^{n}
a_{[i]q}^{(j)}\dfrac{\p}{\p v_{[i]q}^{(j)}}.\label{Z2.23}\eeq
Obviously, $\p_{\bf a}$ and $\p$ commute, i.e.,
\beq \p\p_{\bf a} f=\p_{\bf a}\p f,\quad \mbox{for}\quad f\in\mathcal{D}. \label{Z2.25}\eeq
The set of all vector fields $\p_{\bf a}$ will be denoted by $\mathcal{V}$, which is a Lie algebra
with respect to  the commutator $[\p_{\bf a},\p_{\bf b}]=\p_{\p_{\bf a}{\bf b}-\p_{\bf b}{\bf a}}$.
Let $\Omega^1$ be the dual space of $\mathcal{V}$ consisting of formal $\mathcal{A}$-valued
integral operators
$$X=\dsum_{i=-\infty}^{m-1} \p^{-i-1}X_i,\quad X_i\in\mathcal{A}$$
with the pairing
\beq \la \p_{\bf a}, X\ra=\la {\bf a}, X\ra=\trm\int \res({\bf a}\circ X) dx.\label{Z4.9}\eeq

With the use of the formulae \eqref{Z2.20} and \eqref{Z2.25},
the action of $\mathcal{V}$ on $\mathcal{D}$ can be transferred to $\widetilde{\mathcal{D}}$:
\eqa \p_{\bf a}\tilde{f}=\p_{\bf a}\int f dx=\int \p_{\bf a}f dx
=\dsum_{i=-\infty}^{m-1}\dsum_{q=1}^{n} \int \dfrac{\delta f }{\delta v_{[i]q}}a_{[i]q}dx
=\trm \int \dsum_{i=-\infty}^{m-1} a_i\circ\dfrac{\delta f}{\delta V_i}\, dx.\nn\eeqa
If we set
\beq \dfrac{\delta f}{\delta \mathcal{L}}=\dsum_{i=-\infty}^{m-1} \p^{-i-1}\frac{\delta f}{\delta V_i}
\label{Z2.26}\eeq
and identify the vector ${\bf a}=(a_{n-1},a_{n-2},\cdots)$ with the $\mathcal{A}$-valued $\Psi$DO
${\bf a}=\dsum_{i=-\infty}^{m-1}a_i\p^i,$
we then have
\beq \p_{\bf a}\tilde{f}=\trm\int \res \, ({\bf a} \circ\dfrac{\delta f}{\delta \mathcal{L}}) \, dx,\label{Z4.7} \eeq
which follows
\beq \la \p_{\bf a}, \frac{\delta f}{\delta \mathcal{L}}\ra=\p_{\bf a} \tilde{f}=\la \p_{\bf a},
 d\tilde{f}\ra,\quad d\tilde{f}=\frac{\delta f}{\delta \mathcal{L}} \in {\Omega}^1.\label{Z4.10} \eeq

 \begin{lem}\label{lem3.4} The  mapping
 $\mathcal{H}:\Omega^1 \rightarrow \mathcal{V}$  defined by  $\mathcal{H}(X)=\p_{A^{(z)}(X)}$ is a Hamiltonian
 mapping\footnote{A skew mapping $\mathcal{H}:\Omega^1 \rightarrow \mathcal{V}$
 is said to be Hamiltonian if
(1). $\mathcal{H}\Omega^1 \subset\mathcal{V}$ is a Lie subalgebra;
(2). the 2-form $\omega$  defined by
$ \omega(\mathcal{H}(X), \mathcal{H}(Y))=\la \mathcal{H}(X), Y\ra$
is closed.}, where
 \beq
 A^{(z)}(X)=({\widetilde{\mathcal{L}}}\circ X)_+\circ {\widetilde{\mathcal{L}}}-{\widetilde{\mathcal{L}}}\circ (X\circ {\widetilde{\mathcal{L}}})_+
 \label{eq3.7} \eeq
and $\widetilde{\mathcal{L}}=\mathcal{L}-z$ and $z$ is an arbitrary parameter.
\end{lem}
\begin{proof}When the Frobenius algebra $\mathcal{A}$ is taken to be $\mathbb{R}$, this mapping is the
famous Adler mapping which is a Hamiltonian mapping.  For a general commutative Frobenius algebra, its trace
form is nondengerate and symmetric. We thus follow the same ideas as used in \cite{D1987} to obtain the proof by replacing the scalar
operators by $\mathcal{A}$-valued operators. \end{proof}

We rewrite $A^{(z)}(X)$ in \eqref{eq3.7} as
$$A^{(z)}(X)=H^{m(0)}(X)+z\,H^{m(\infty)}(X),$$
that is to say, \beq H^{m(0)}(X)=({\mathcal{L}}\circ X)_+\circ {\mathcal{L}}-{\mathcal{L}}\circ (X\circ {\mathcal{L}})_+,
\quad H^{m(\infty)}(X)=[{\mathcal{L}}_{-},X_+]_{-}-[\mathcal{L}_{+},X_{-}]_{+}.\label{Z2.17}\eeq
By using Lemma \ref{lem3.4}, $H^{m(0)}$ and $H^{m(\infty)}$ are Hamiltonian mappings. We thus get
two compatible Poisson brackets of the $\mathcal{A}$-KP
 hierarchy associated with ${\mathcal{L}}:=L^m$ are given by
\eqa \left\{\tilde{f},\tilde{g}\right\}^{m(\infty)}
&=&\trm \int\res H^{m(\infty)}\left(\frac{\delta f}{\delta \mathcal{L}}\right)\circ\frac{\delta g}{\delta \mathcal{L}}\,dx
\label{Z2.18}\\
&=&\trm \int\res \left( \left[{\mathcal{L}}_{-},\left(\frac{\delta f}{\delta \mathcal{L}}\right)_+\right]_{-}
-\left[\mathcal{L}_{+},\left(\frac{\delta f}{\delta \mathcal{L}}\right)_{-}\right]_{+}\right)\circ\frac{\delta g}{\delta \mathcal{L}}\,dx
\nn\eeqa
and
\eqa\left\{\tilde{f},\tilde{g}\right\}^{m(0)}
&=&\trm \int\res H^{m(0)}\left(\frac{\delta f}{\delta \mathcal{L}}\right)\circ
\frac{\delta g}{\delta \mathcal{L}}\,dx \label{Z2.19} \\
&=&\trm \int\res \left(\left({\mathcal{L}}\circ\frac{\delta f}{\delta \mathcal{L}}\right)_+\circ {\mathcal{L}}
-{\mathcal{L}}\circ\left(\frac{\delta f}{\delta \mathcal{L}}\circ{\mathcal{L}}\right)_+\right)\circ
\frac{\delta g}{\delta \mathcal{L}}\,dx,\nn\eeqa
where $\tilde{f}$, $\tilde{g}$ are two functionals. Furthermore, we have

\begin{thm}\label{thm3.4}The $\mathcal{A}$-KP hierarchy $\dfrac{\p L}{\p t_r}=[B_r,L]$ admits a bi-Hamiltonian representation given by
\beq \dfrac{\p \mathcal{L}}{\p t_r}=H^{m(0)}\left(\frac{\delta h_r}{\delta \mathcal{L}}\right)=H^{m(\infty)}\left(\frac{\delta g_r}{\delta \mathcal{L}}\right) \eeq
with the Hamiltonians
 $$  \tilde{h}_r=\frac{m}{r}\trm\int \res L^{r}\, dx\,\quad\mbox{and}\quad
 \tilde{g}_r=-\frac{m}{r+m} \trm \int \res L^{m+r} \, dx. $$
\end{thm}
\begin{proof}Observe that the $\mathcal{A}$-KP hierarchy $\dfrac{\p L}{\p t_r}=[B_r,L]$ is equivalent to
$\dfrac{\p \mathcal{L}}{\p t_r}=[B_r,\mathcal{L}]$. By definition in \eqref{Z2.20},
one obtains
$$\dfrac{\delta}{\delta\mathcal{L}} \trm\int \res L^r dx=\dfrac{r}{m}L^{r-m}.$$
With the help of \eqref{Z2.17}, one thus gets
$$H^{m(0)}(\frac{\delta h_r}{\delta \mathcal{L}})=H^{m(0)}(L^{r-m})=[B_r,\mathcal{L}]$$
and
\eqa H^{m(\infty)}(\frac{\delta g_r}{\delta \mathcal{L}})&=&-H^{m(\infty)}(L^{r})=[\mathcal{L}_{+},L^{r}_{-}]_{+}-[{\mathcal{L}}_{-},L^{r}_+]_{-} \nn\\
&=&[\mathcal{L},L^{r}_{-}]_{+}-[{\mathcal{L}},L^{r}_+]_{-}=[B_r,\mathcal{L}],\nn\eeqa
which yields this theorem.
\end{proof}

\subsection{Case $U_0=0$, i.e., $V_{m-1}=0$ }If we restrict to $V_{m-1}=0$, it is easy to check that
the first Hamiltonian structure automatically reduces to this submanifold,
but the second one is reducible if and only if
\beq \res [\mathcal{L}, \frac{\delta f}{\delta \mathcal{L}}]=0.\label{Z2.29}\eeq
which is equivalent to the condition
\beq  X_{m-1}=\dfrac{1}{m}\dsum_{i=-\infty}^{m-2}\left(\left(\begin{array}{r}-i-1\\m-i\end{array}\right)
X_i^{(m-i-1)}+\dsum_{j=i+1}^{m-1}
\left(\begin{array}{r}-i-1\\j-i\end{array}\right)(X_i\circ V_j)^{(j-i-1)}\right)\,,\label{Z2.30} \eeq
where $X_i=\dfrac{\delta f}{\delta V_i}\in \mathcal{A}$.
We denote the corresponding reduced brackets by $\{~,~\}^{m(\infty)}$ and $\{~,~\}_D^{m(0)}$.

\begin{cor}\label{cor4.5} The
coupled KP hierarchy defined in \cite{CO2006} has at least $n$ ``basic" different local
bi-Hamiltonian structures.\end{cor}
\begin{proof}As explained in the introduction, the coupled KP hierarchy defined in \cite{CO2006,L2007} is exactly the
$\mathcal{Z}_n$-KP hierarchy.  According to Example \ref{ex2.3}, the algebra $\mathcal{Z}_n$ has at least $n$-``basic"
different ways to be realized as the Frobenius algebra.  With this,
the corollary follows immediately from Theorem \ref{thm3.4}.
  \end{proof}

\begin{de}In terms of the basis $\{v_{[i]q}\}$, the second
Poisson bracket $\{~,~\}^{m(0)}$
for $L^m$ in \eqref{Z2.15} and the reduced bracket $\{~,~\}_D^{m(0)}$ for
$L^m$ with the constraint $V_{m-1}=0$ will provide
two kinds of local \mbox{W}-type algebras, we call them the {$\mbox{W}^{(n,m)}_{\mathcal{A}\mathbb{KP}}$-algebra} and the
{$\mbox{W}^{(n,m)}_{\infty}$-algebra} respectively. Under the reduction $L^m_{-}=0$, the
corresponding algebras are called the {$\mbox{W}^{(n,m)}_{\mathcal{A}\mathbb{GD}}$-algebra} and the
{$\mbox{W}_{(n,m)}$-algebra} respectively.
\end{de}

With the use of \eqref{Z2.15} and \eqref{Z2.30}, one knows that all of them are
{\bf local matrix} generalizations of $W$-algebras. To conclude this section, two examples
will be given to illustrate our construction.

\begin{ex} Consider the $\mathcal{A}$-$\mbox{KdV}$ hierarchy with the Lax operator $L^2=\be \p^2+V$,
i.e., $L^2_{-}=0$. We denote $X=\p^{-2}X_1+\p^{-1}X_0$ and $Y=\p^{-2}Y_1+\p^{-1}Y_0.$
The condition \eqref{Z2.29} becomes $X_1=\dfrac{1}{2}X_0'$, then we have
\beq H^{2(\infty)}=[X,L^2]_+=-2X_0'\nn\eeq
and
\beq H^{2(0)}(X)=({L^2}\circ X)_+\circ {L^2}-{L^2}\circ(X\circ {L^2})_+
=2V\circ X_0'+X_0\circ V'+\frac{1}{2}X_0''' .\nn\eeq
Thus two compatible Poisson brackets of the $\mathcal{A}$-$\mbox{KdV}$ hierarchy (\cite{Zuo-2014}) are given by
\beq  \left\{\tilde{f},\tilde{g}\right\}^{2(\infty)}=2\, \trm
\int \frac{\delta f}{\delta V}\circ
\frac{\p}{\p x}\frac{\delta g}{\delta V}dx\nn \eeq
and
\beq
\left\{\tilde{f},\tilde{g}\right\}_D^{2(0)}=-\frac{1}{2}\,\trm
 \int \frac{\delta f}{\delta V}\circ
\left(\be\frac{\p^3}{\p x^3}+2V\frac{\p}{\p x}
+2\frac{\p}{\p x}V\right)\circ\frac{\delta g}{\delta V}dx.\nn\eeq

In particular, if one chooses the algebra $\mathcal{A}$ to be the algebra $\mathcal{Z}_2$
defined in Example \ref{ex2.3}, one obtains the $\mathcal{Z}_2$-KdV equation for $V=ve_1+we_2$
given by
\beq
4v_{t}-12vv_{x}-v_{xxx}=0,\quad 4w_{t}-12(vw)_x-w_{xxx}=0.
 \label{SZ4.21} \eeq
According to Corollary \ref{cor4.5}, the system \eqref{SZ4.21} can be written as
\beq \left(
\begin{array}{c}
v\\
w\end{array}\right)_t=\left(
\begin{array}{cc}
0& \p\\
\p& 0
\end{array}\right)\left(
\begin{array}{c}
\frac{\delta H_{2}}{\delta v}\\
\frac{\delta H_{2}}{\delta w}\end{array}\right)=\left(
\begin{array}{cc}
0& J_0\\
J_0 &  J_1
\end{array}\right)\left(
\begin{array}{c}
\frac{\delta H_{1}}{\delta v}\\
\frac{\delta H_{1}}{\delta w}\end{array}\right) \nn \eeq
with Hamiltonians
 $$H_1=\int_{\mathbb{S}^1}vw dx,\quad H_2= \int_{\mathbb{S}^1}(\frac{3}{2}v^2w+\frac{1}{4}vw_{xx}) dx;$$
and
\beq \left(
\begin{array}{c}
v\\
w\end{array}\right)_t=\left(
\begin{array}{cc}
0 & \p\\
\p & -\p
\end{array}\right)\left(
\begin{array}{c}
\frac{\delta \widetilde{H}_{2}}{\delta v}\\
\frac{\delta \widetilde{H}_{2}}{\delta w}\end{array}\right)=\left(
\begin{array}{cc}
0& J_0\\
J_0 & J_1-J_0
\end{array}\right)\left(
\begin{array}{c}
\frac{\delta \widetilde{H}_{1}}{\delta v}\\
\frac{\delta \widetilde{H}_{1}}{\delta w}\end{array}\right) \nn \eeq
with Hamiltonians
 $$\widetilde{H}_1= \int_{\mathbb{S}^1}(\frac{1}{2} v^2+vw) dx,\quad
 \widetilde{H}_2=\int_{\mathbb{S}^1}(\frac{3}{2}v^2w+\frac{1}{4}vw_{xx}+\frac{1}{2}v^3+\frac{1}{8}vv_{xx}) dx,$$
where $J_0=\dfrac{1}{4}\p^3+v\p+\p v$ and $J_1=w\p+\p w$.\end{ex}

\begin{ex}\label{ex2.8}{\bf$[$The $\mathcal{A}$-$\mbox{Boussinesq}$ hierarchy$]$}
In this case, we have $\mathcal{L}=\mathrm{I}_m\p^3+V_1\p+V_0$. Let us take
$\tilde{f},\, \tilde{g} \in\widetilde{\mathcal{D}}$,
 and denote
 $$ X_j=\dfrac{\delta{f}}{\delta V_j},\quad Y_j=\dfrac{\delta{g}}{\delta V_j},\quad j=0,\,1.$$
 Using the condition \eqref{Z2.30}, we have
 \beq  \dfrac{\delta{f}}{\delta\mathcal{L}}=\p^{-3}X_2+\p^{-2}X_1+\p^{-1}X_0,
\quad \dfrac{\delta{g}}{\delta \mathcal{L}}=\p^{-3}Y_2+\p^{-2}Y_1+\p^{-1}Y_0.\nn\eeq
where
$X_2=X_1'-\frac{1}{3}X_0''-\frac{1}{3} X_0V_1$ and $Y_2=Y_1'-\frac{1}{3}Y_0''-\frac{1}{3} Y_0V_1$.

A direct calculation gives two Poisson brackets of the $\mathcal{A}$-$\mbox{Boussinesq}$ hierarchy
\beq  \left\{\tilde{f},\tilde{g}\right\}^{3(\infty)}=3 \, \trm \int (X_1Y_0'+X_0Y_1')dx\nn \eeq
and
\eqa
\left\{\tilde{f},\tilde{g}\right\}_D^{3(0)}
&=&\trm \int \left(\frac{2}{3}X_0Y_0^{(5)}-X_0Y_1^{(4)}+X_1Y_0^{(4)}-2X_1Y_1^{(3)}\right)dx \nn\\
&+& \trm \int  \left(\dfrac{1}{3}X_0Y_0'-\dfrac{1}{3}X_0'Y_0\right)V_1^2 \,dx\nn\\
&+& \trm \int  \left(\dfrac{2}{3}X_0Y_0^{(3)}-\dfrac{2}{3}X_0^{(3)}Y_0+X_1''Y_0-X_0Y_1''
+X_1'Y_1-X_1Y_1'\right)V_1 \,dx\nn\\
&+& \trm \int \left(X_0Y_0''-X_0''Y_0+2X_1'Y_0-X_1Y_0'+X_0'Y_1-2X_0Y_1'\right)V_0 \,dx.\nn
\eeqa
More specifically, by analogy to the classical $W$-algebra in \cite{Baka1989, FIZ1991}, we set
\beq W_2=V_1,\quad W_3=V_0-\frac{1}{2}V_1',\nn \eeq
then for any two $\mathcal{A}$-valued test functions $F$ and $G$, we have
\eqa && \left\{\trm\int F W_2 dx\, ,\,\trm\int G W_2 dx\right\}_D^{3(0)}
=\trm \int \left(2F^{(3)}+{\bf 2}W_2F'+W_2'F\right)G\,dx,\nn\eeqa
and
 \eqa &&\left\{\trm\int F W_2 dx\, ,\,\trm\int G W_3 dx\right\}_D^{3(0)}
= \trm \int \left({\bf 3}W_3F'+W_3'F\right)G\,dx,\nn \eeqa
and
\eqa && \left\{\trm\int F W_3 dx\, ,\,\trm\int G W_3 dx\right\}_D^{3(0)}=\frac{1}{6}\, \trm \int
\left((2\,FG'-2\,F'G)W_2^2\right.\nn\\  &+& \left. FG^{(5)}\right) \,dx
+\frac{1}{12}\,\trm \int \left(2\,FG^{(3)}-2\,F^{(3)}G+3F''G'-3F'G''\right) W_2 dx.\nn \eeqa
We thus confirm that $W_{\bf k}$ for ${\bf k}=2,3$ are spin-${\bf k}$ conformally primary
$\mathcal{A}$-valued fields. But notice that the equation $\trm F W_2^2=(\trm FW_2)^2$ has no
$\mathcal{A}$-valued non-zero solution, which means the classical $\mbox{W}_3$-algebra is not
a subalgebra of the $\mbox{W}_{(n,3)}$-algebra for ${\rm{dim}}\mathcal{A}=n>1$.
\end{ex}

\section{The dispersionless $\mathcal{A}$-KP hierarchy}
Because of the similarities in the theories of dispersionless and dispersive KP equations (see \cite{Ian-1995,Zuo-2006}),
we list here the analogous results for the $\mathcal{A}$-$\mbox{dKP}$
hierarchy without proofs. We will use the following notation in this part. For an
$\mathcal{A}$-valued Laurent series of the form
$A=\dsum_{i} A_ip^i $, we denote by $A_+$ the polynomial part of the Laurent series $A$ and
$A_{-}=A-A_+$, $\res(A)=a_{-1}$. Let \beq
L=\be p+U_1 p^{-1}+U_2 p^{-2}+\cdots,
\label{Z3.1}\eeq be an $\mathcal{A}$-valued Laurent series.

\begin{de}
The $\mathcal{A}$-$\mbox{dKP}$ hierarchy is the set of equations of motion
\beq \frac{\p L}{\p t_r}=\{L^r_+,L\},\label{Z5.2}\eeq
where $\{~,~\}$ is defined by
$\{A,B\}=\frac{\p A}{\p p}\circ\frac{\p B}{\p x}-\frac{\p A}{\p x}\circ\frac{\p B}{\p p}.$
\end{de}

Let us assume that  $L^m$, $m\in\mathbb{N}$, is of the form
\beq \mathcal{L}:=L^m=\be p^m+V_{m-1}p^{m-1}+\cdots. \label{Z5.3}\eeq
Taking a dispersionless limit of Hamiltonian structures for the $\mathcal{A}$-\mbox{KP} hierarchy,
we get the first and the second Poisson  brackets of the $\mathcal{A}$-\mbox{dKP} hierarchy associated
with $\mathcal{L}$ in \eqref{Z5.3} as follows
\eqa \left\{\tilde{f},\tilde{g}\right\}^{m(\infty)}
=\trm \int\res \left( \left\{{\mathcal{L}}_{-},(\frac{\delta f}{\delta \mathcal{L}})_+\right\}_{-}
-\left\{\mathcal{L}_{+},(\frac{\delta f}{\delta \mathcal{L}})_{-}\right\}_{+}\right)\circ
\frac{\delta g}{\delta \mathcal{L}}\,dx
\label{Z5.5}\eeqa
and
\eqa\left\{\tilde{f},\tilde{g}\right\}^{m(0)}
=\trm \int\res \left(({\mathcal{L}}\circ\frac{\delta f}{\delta \mathcal{L}})_+\circ{\mathcal{L}}
-{\mathcal{L}}\circ(\frac{\delta f}{\delta \mathcal{L}}\circ{\mathcal{L}})_+\right)\circ
\frac{\delta g}{\delta \mathcal{L}}\,dx,\label{Z5.6}\eeqa
where $\tilde{f}$, $\tilde{g}\in\tilde{\mathcal{D}}$ are two functionals.
The variational derivative $\dfrac{\delta f}{\delta \mathcal{L}}$
is given by
\beq \frac{\delta f}{\delta \mathcal{L}}=\dsum_{i=-\infty}^{m-1}\frac{\delta f}{\delta V_i}p^{-i-1},
 \label{Z5.7}\eeq
 where $\frac{\delta f}{\delta V_i}$ is defined in \eqref{Z2.20}.  When we restrict
 these to the submanifold $V_{m-1}=0$, the first Hamiltonian structure automatically reduces
 to this submanifold, but the second one is reducible if and only if
\beq \res \left\{\mathcal{L}, \frac{\delta f}{\delta \mathcal{L}}\right\}=0.\label{Z5.8}\eeq
Similarly, in terms of the basis $\{v_{[i]q}\}$, the second Poisson bracket $\{~,~\}^{m(0)}$
for $L^m$ in \eqref{Z5.3} and the reduced bracket $\{~,~\}^{m(0)}_D$ for $L^m$ with the constraint
$V_{m-1}=0$ will provide two kinds of local $w$-type algebras.

\section{Conclusions}
In summary we have introduced the Frobenius algebra-valued KP hierarchy and studied
the existence of $\tau$-functions and Hamiltonian structures. Regarding scalar fields as components
of a more basic $\mathcal{A}$-valued field is a more elegant approach: it is not basis dependent and it
automatically stresses the algebraic properties more clearly. Other properties can then be traced back, for example, to the freedom in
the definition of the Frobenius form.
Via the properties of the second Hamiltonian
structures, we have obtained some local matrix generalizations of $W$-algebras.
An interesting byproduct is that the coupled KP hierarchy in \cite{CO2006} has at least
$n$-``basic" different local bi-Hamiltonian structures. The methods in the paper may clearly be applied to other theories of a similar type which have an underlying Lax equation,  for example, Toda-hierarchies and reductions of these theories \cite{ZZ-2015,TZ-2015}.

In a separate paper $\mathcal{A}$-valued Frobenius manifolds, topological quantum field theories and bi-Hamiltonian structures are constructed \cite{SZ1}. These constructions are different in character to those in this paper: they are developed without any use of Lax equations, relying on a \lq lifting\rq~construction from scalar to algebra-valued fields. There will, clearly, be an overlap, with the theory of $\mathcal{A}$-valued KdV and dKdV equations being the most obvious example.

\medskip
\noindent{\bf Acknowledgements.}
D.Zuo is partially supported by NCET-13-0550 and NSFC (11271345, 11371338) and the Fundamental Research
Funds for the Central Universities.
\medskip



\begin{thebibliography}{99}



\bibitem{A1979} M.Adler, \emph{On a trace functional for formal pseudodifferential
operators and the symplectic structure of the Korteweg-de Vries type equations.}
Invent. Math. 50 (1979)219--248.

\bibitem{Baka1989} I.Bakas, \emph{Higher Spin Fields and the Gelfand-Dickey Algebra.}
Commun. Math. Phys., 123 (1989) 627--639.

\bibitem{Bial1995} A.Bilal, \emph{Non-Local Matrix Generalizations of W-Algebras.}
 Commu. Math. Phys., 170 (1995) 117--150.

\bibitem{CO2006}P.Casati and G.Ortenzi,\emph{New integrable hierarchies from
vertex operator representations of polynomial Lie algebras.} J.Geom.Phys.
56(2006)418--449.



\bibitem{DJKM1983}E. Date, M. Jimbo, M. Kashiwara, T. Miwa, \emph{Transformation
Groups for Soliton Equations.} In: Proceedings of R.I.M.S. Symposium on
Nonlinear Integrable Systems--Classical Theory and Quantum Theory,
World Scientific, (1983)39--119.

\bibitem{D1987}L.A.Dickey, \emph{On Hamiltonian and Lagrangian formalisms for the KP hierarchy of integrable
equations.} Ann. N.Y. Acad. Sci. 491 (1987), 131--148.

\bibitem{D1997}L.A.Dickey, \emph{Lectures on classical W-algebras.}  Acta Appl.Math. 47(1997)243--321.


\bibitem{Dick2003}L.A.Dickey,  \emph{Soliton equations and Hamiltonian systems.} Second edition.
Advanced Series in Mathematical Physics, 26. World Scientific 2003.

\bibitem{Du} B.Dubrovin, \emph{Geometry of $2$D topological field theories.}
Lecture Notes in Mathematics 1620,(1996)120--348 Springer.

\bibitem{FIZ1991}P.Di Franscessco, C.Itzykson and J.B.Zuber, \emph{Classical W-algebras.}
Commun. Math. Phys. 140 (1991) 543--567.

\bibitem{FR1989}A.P.Fordy, A.G.Reyman and M.A.Semenov-Tian-Shansky,
\emph{Classical $r$-matrices and compatible Poisson brackets for coupled KdV
systems}, Lett. Math. Phys. \textbf{17} (1989) 25--29.

\bibitem{GD1978} I.M.Gelfand and L.A.Dickey, \emph{Family of Hamiltonian structures connected with
integrable non-linear equations}, Preprint, IPM, Moscow (in Russian), 1978. English version in:
Collected papers of I.M. Gelfand, vol. 1, Springer-Verlag (1987), 625-46.

\bibitem{HXT2003} R. Hirota, X.B.Hu and X.Y.Tang,
\emph{A vector potential KdV equation
and vector Ito equation: soliton solutions, bilinear B\"acklund
transformations and Lax pairs.}
J. Math. Anal. Appl.  288  (2003),  no. 1, 326--348.



\bibitem{Kup2000}B.A. Kupershmidt,\emph{KP or mKP: Noncommutative Mathematics of Lagrangian, Hamiltonian, and
Integrable Systems.}  Mathematical Surveys and Monographs 78, American Mathematical Society 2000.

\bibitem{L2007}J. van de Leur, \emph{B\"acklund transformations for new integrable hierarchies
related to the polynomial Lie algebra $gl_\infty^{(n)}$.} J.Geom.Phys. 57(2007)435--447.


\bibitem{MF1996} W.X.Ma and B.Fuchssteinery,
\emph{The bihamiltonian structure of the perturbation equations of KdV Hierarchy.}
Phys. Lett. A 213 (1996) 49--55.

\bibitem{OS} P.J. Olver and V.V. Sokolov, \emph{Integrable evolution equations on associative algebras},
Commun. Math. Phys., 2(1998) 245--268.

\bibitem{R1987}A.O.Radul,\emph{Two series of Hamiltonian structures for the hierarchy of
Kadomtsev-Petviashvili equations}. in Mironov, Moroz and Tshernyatin (eds), Applied Methods
of Nonlinear Analysis and Control, Moscow, State Univ., 1987, pp. 149--157.

\bibitem{SS1982} M. Sato, \emph{Soliton equations as dynamical systems on infinite dimensional
 Grassmann manifolds}. RIMS Kokyuroku 439 (1981) 30--46.

 \bibitem{Ian-1995} I.A.B. Strachan, \emph{ The Moyal bracket and the dispersionless limit of the KP hierarchy},
 J. Phys. A: Math. Gen., 28(1995)1967--75.

 \bibitem{SZ1} I.A.B. Strachan and D. Zuo, \emph{ Frobenius manifolds and Frobenius algebra-valued
integrable systems}, preprint arXiv:1403.0021.

\bibitem{W1983} Y.Watanabe, \emph{Hamiltonian structure of Sato's hierarchy of KP equations and a coadjoint orbit
of a certain formal Lie group}. Lett. Math. Phys. 7(2) (1983), 99--106.

\bibitem{Zuo-2006}D.Zuo, \emph{On the Kuperschmidt-Wilson theorem for the
Moyal-Kadomtsev-Petviasfvili hierarchy.} Inverse Problems 22(2006)1959--1966.

\bibitem{Zuo-2013} D.Zuo, \emph{Local matrix generalizations of W-algebras, arXiv:1401.2216v1.} (Unpublished)

\bibitem{Zuo-2014} D.Zuo,\emph{The Frobenius-Virasoro algebra and Euler equations.} J. Geom. Phys., 86 (2014)203--210.

\bibitem{ZZ-2015} H.Zhang and D.Zuo, \emph{Hamiltonian structures of the constrained $\mathcal{F}$-valued KP hierarchy.}
Rep.Math.Phys., 76 (2015) 116--129.

\bibitem{TZ-2015} K-L.Tian and D.Zuo, In preparation.






\end{thebibliography}
\end{document}